\documentclass{IEEEtran}
\usepackage{graphicx,color}
\usepackage{textcomp}
\usepackage{hyperref}
\hypersetup{hidelinks=true}

\usepackage{cite}
\usepackage{amsmath,amssymb,amsfonts}
\usepackage{algorithm}
\usepackage{algpseudocode}
\usepackage{booktabs}
\usepackage{graphicx}
\usepackage{textcomp}
\usepackage{amsthm}
\usepackage{xcolor}
\usepackage{orcidlink}
\usepackage[acronym,shortcuts]{glossaries}

\usepackage{amsthm}
\theoremstyle{remark}
\newtheorem{remark}{Remark}
\theoremstyle{plain}

\newtheorem{lemma}{Lemma}
\newtheorem{proposition}{Proposition}
\newtheorem{theorem}{Theorem}
\newtheorem{corollary}{Corollary}
\def\BibTeX{{\rm B\kern-.05em{\sc i\kern-.025em b}\kern-.08em
T\kern-.1667em\lower.7ex\hbox{E}\kern-.125emX}}
\newacronym{PCA}{PCA}{principal component analysis}
\newacronym{MDSM}{MDSM}{multi-domain sparse modulation}
\newacronym{P2P}{P2P}{point-to-point}
\newacronym{OTAC}{AirComp}{over-the-air computing}
\newacronym{TX}{TX}{transmitter}
\newacronym{RX}{RX}{receiver}
\newacronym{IoT}{IoT}{Internet of Things}
\newacronym{AI/ML}{AI/ML}{artifitial intelligence/machine learning}
\newacronym{SDR}{SDR}{semi-definite relaxation}
\newacronym{EVD}{EVD}{eigenvalue decomposition}
\newacronym{GR}{GR}{Gaussian randomization}
\newacronym{SCA}{SCA}{successive convex approximation}
\newacronym{BnB}{BnB}{branch and bound}
\newacronym{QT}{QT}{quadratic transform}
\newacronym{RQ}{RQ}{Rayleigh quotient}
\newacronym{SOCP}{SOCP}{second-order cone programming}
\newacronym{CDF}{CDF}{cumulative distribution function}
\newacronym{UF}{UF}{uniform-forcing}
\newacronym{AP}{AP}{access point}
\newacronym{RSDR}{R-SDR}{regularized semi-definite relaxation}
\newacronym{R-SDR}{R-SDR}{regularized SDR}
\newacronym{flops}{flops}{floating point operations}
\newacronym{ED}{ED}{edge device}
\newacronym{SINR}{SINR}{signal to interference-plus-noise ratio}
\newacronym{SIC}{SIC}{successive interference cancellation}
\newacronym{CSI}{CSI}{channel state information}
\newacronym{LoS}{LoS}{line-of-sight}
\newacronym{NLoS}{NLoS}{non-LoS}
\newacronym{RPE}{RPE}{radar parameter estimation}
\newacronym{OTFS}{OTFS}{orthogonal time frequency space}
\newacronym{AFDM}{AFDM}{affine frequency division multiplexing}
\newacronym{CRLB}{CRLB}{Cram{\`e}r-Rao lower bound}
\newacronym{BCRLB}{BCRLB}{Bayesian Cram{\`e}r-Rao lower bound}
\newacronym{BBI}{BBI}{Bayesian bilinear inference}
\newacronym{AoA}{AoA}{angle-of-arrival}
\newacronym{SNR}{SNR}{signal-to-noise ratio}
\newacronym{ML}{ML}{maximum likelihood}
\newacronym{MIMO}{MIMO}{multiple-input multiple-output}
\newacronym{MISO}{MISO}{multiple-input single-output}
\newacronym{SIMO}{SIMO}{single-input multiple-output}
\newacronym{SISO}{SISO}{single-input single-output}
\newacronym{MUSIC}{MUSIC}{multiple signal classification}
\newacronym{MU}{MU}{multi-user}
\newacronym{ROOT-MUSIC}{ROOT-MUSIC}{ROOT multiple signal classification}
\newacronym{JCAS}{JCAS}{joint communication and sensing}
\newacronym{JCR}{JCR}{joint communications and radar}
\newacronym{ISAC}{ISAC}{integrated sensing and communications}
\newacronym{3D}{3D}{three-dimensional}
\newacronym{2D}{2D}{two-dimensional}
\newacronym{1D}{1D}{one-dimensional}
\newacronym{BF}{BF}{beamforming}
\newacronym{ROI}{ROI}{region of interest}
\newacronym{mmWave}{mmWave}{millimeter-wave}
\newacronym{MF}{MF}{matched-filter}
\newacronym{DD}{DD}{delay-Doppler}
\newacronym{SotA}{SotA}{state-of-the-art}
\newacronym{ULA}{ULA}{uniform linear array}
\newacronym{QAM}{QAM}{quadrature amplitude modulation}
\newacronym{ISFFT}{ISFFT}{inverse symplectic finite Fourier transform}
\newacronym{SFFT}{SFFT}{symplectic finite Fourier transform}
\newacronym{ISI}{ISI}{inter-symbol interference}
\newacronym{AWGN}{AWGN}{additive white Gaussian noise}
\newacronym{MSE}{MSE}{mean-squared-error}
\newacronym{LMMSE}{LMMSE}{linear minimum mean square error}
\newacronym{RMSE}{RMSE}{root mean square error}
\newacronym{ESPRIT}{ESPRIT}{estimation of signal parameters via rotational invariant techniques}
\newacronym{OFDM}{OFDM}{orthogonal frequency division multiplexing}
\newacronym{OCDM}{OCDM}{orthogonal chirp division multiplexing}
\newacronym{BS}{BS}{base station}
\newacronym{UE}{UE}{user equipment}
\newacronym{JCEDD}{JCEDD}{joint channel estimation and data detection}
\newacronym{PDA}{PDA}{probabilistic data association}
\newacronym{PMF}{PMF}{probability mass function}
\newacronym{PBiGaBP}{PBiGaBP}{parametric bilinear Gaussian belief propagation}
\newacronym{PBiGAMP}{PBiGAMP}{parametric bilinear generalized approximate message passing}
\newacronym{GaBP}{GaBP}{Gaussian belief propagation}
\newacronym{FT}{FT}{frequency-time}
\newacronym{DFT}{DFT}{discrete Fourier transform}
\newacronym{IDFT}{IDFT}{inverse discrete Fourier transform}
\newacronym{TD}{TD}{time domain}
\newacronym{wlg}{w.l.g.}{without loss of generality}
\newacronym{CP}{CP}{cyclic prefix}
\newacronym{DAF}{DAF}{discrete affine Fourier}
\newacronym{DAFT}{DAFT}{discrete affine Fourier transform}
\newacronym{IDAFT}{IDAFT}{inverse discrete affine Fourier transform}
\newacronym{CPP}{CPP}{\textit{chirp-periodic} prefix}
\newacronym{IDZT}{IDZT}{inverse discrete Zak transform}
\newacronym{DZT}{DZT}{discrete Zak transform}
\newacronym{P/S}{P/S}{parallel-to-serial}
\newacronym{S/P}{S/P}{serial-to-parallel}
\newacronym{SBL}{SBL}{sparse Bayesian learning}
\newacronym{MPA}{MPA}{message passing algorithms}
\newacronym{EM}{EM}{expectation maximization}
\newacronym{sIC}{soft IC}{soft interference cancellation}
\newacronym{soft RG}{soft RG}{soft replica generation}
\newacronym{BG}{BG}{belief generation}
\newacronym{SGA}{SGA}{scalar Gaussian approximation}
\newacronym{CLT}{CLT}{central limit theorem}
\newacronym{PDF}{PDF}{probability density function}
\newacronym{QPSK}{QPSK}{quadrature phase-shift keying}
\newacronym{ICI}{ICI}{inter-carrier interference}
\newacronym{BER}{BER}{bit error rate}
\newacronym{DoF}{DoF}{degrees-of-freedom}
\newacronym{VGA}{VGA}{vector Gaussian approximation}
\newacronym{FD}{FD}{full-duplex}
\newacronym{NMSE}{NMSE}{normalized mean square error}
\newacronym{KL}{KL}{Kullback-Leibler}
\newacronym{LASSO}{LASSO}{least absolute shrinkage and selection operator}
\newacronym{FP}{FP}{fractional programming}
\newacronym{CC}{CC}{communication-centric}
\newacronym{RC}{RC}{raised-cosine}
\newacronym{RRC}{RRC}{root raised-cosine}
\newacronym{6G}{6G}{sixth-generation}
\newacronym{V2X}{V2X}{vehicle-to-everything}
\newacronym{LEO}{LEO}{low-earth orbit}
\newacronym{I/O}{I/O}{input-output}
\newacronym{CE}{CE}{channel estimation}
\newacronym{ICC}{ICC}{integrated communication and computing}
\newacronym{ISCC}{ISCC}{integrated sensing, communications and computing}
\newacronym{PAM}{PAM}{pulse amplitude modulation}
\newacronym{iid}{i.i.d.}{independent and identically distributed}
\newacronym{MEC}{MEC}{mobile edge computing}
\newacronym{REMS}{REMS}{reconfigurable electromagnetic structure}
\newacronym{RIS}{RIS}{reconfigurable intelligent surface}
\newacronym{MMSE}{MMSE}{minimum mean square error}
\newacronym{DPC}{DPC}{dirty paper coding}
\newacronym{GAN}{GAN}{generative adversarial network}
\newacronym{LLR}{LLR}{log-likelihood ratio}
\newacronym{SPRT}{SPRT}{sequential probability ratio test}
\newacronym{HMM}{HMM}{hidden Markov model}
\newacronym{EMA}{EMA}{exponential moving averages}
\newacronym{AR-PLA}{AR-PLA}{adversarially robust physical layer authentication}
\newacronym{5G}{5G}{fifth generation}
\newacronym{PLA}{PLA}{physical layer authentication}
\newacronym{AUC}{AUC}{area under the curve}

\begin{document}
% \receiveddate{XX Month, XXXX}
% \reviseddate{XX Month, XXXX}
% \accepteddate{XX Month, XXXX}
% \publisheddate{XX Month, XXXX}
% \currentdate{XX Month, XXXX}
% \doiinfo{XXXX.2022.1234567}

% \markboth{}{Author {et al.}}

\title{Adversarially Robust MIMO Physical Layer Authentication for Non-Stationary  Channels}

\author{
    Ali Khandan Boroujeni\textsuperscript{\orcidlink{0009-0003-5007-8535}},~\IEEEmembership{Graduate Student Member,~IEEE,}
    Ghazal Bagheri\textsuperscript{\orcidlink{0009-0006-2740-8235}},~\IEEEmembership{Graduate Student Member,~IEEE,} \\
    Kuranage Roche Rayan Ranasinghe\textsuperscript{\orcidlink{0000-0002-6834-8877}},~\IEEEmembership{Graduate Student Member,~IEEE,} \\
    Giuseppe Thadeu Freitas de Abreu\textsuperscript{\orcidlink{0000-0002-5018-8174}},~\IEEEmembership{Senior Member, IEEE,}
    Stefan Köpsell\textsuperscript{\orcidlink{0000-0002-0466-562X}},~\IEEEmembership{Senior Member, IEEE,} \\
    and
    Rafael F. Schaefer\textsuperscript{\orcidlink{0000-0002-1702-9075}},~\IEEEmembership{Senior Member, IEEE}
    % \vspace{-2ex}
    \thanks{Ali Khandan Boroujeni, Stefan Köpsell, and Rafael F. Schaefer are with the Barkhausen Institut and Technische Universit\"at Dresden, 01067 Dresden, Germany (emails: ali.khandanboroujeni@barkhauseninstitut.org; \{stefan.koepsell,rafael.schaefer\}@tu-dresden.de).}
    \thanks{Ghazal Bagheri is with Technische Universit\"at Dresden, 01187 Dresden, Germany (email: ghazal.bagheri@tu-dresden.de).}
    \thanks{Kuranage Roche Rayan Ranasinghe and Giuseppe Thadeu Freitas de Abreu are with the School of Computer Science and Engineering, Constructor University (previously Jacobs University Bremen), Campus Ring 1, 28759 Bremen, Germany (emails: \{kranasinghe,gabreu\}@constructor.university).}
    
}

\maketitle

\begin{abstract}
We propose an \ac{AR-PLA} framework tailored for non-stationary \ac{MIMO} wireless channels. 
The framework integrates sequential Bayesian decision-making, deep feature extraction via contrastive learning, and generative adversarial modeling to simulate adaptive spoofers. 
Unlike conventional methods that assume stationary channels or independent observations, our approach explicitly accounts for temporal and spatial correlations, \ac{LoS} blockages, and dynamic spoofing strategies. 
A comprehensive analytical characterization of the authentication performance using both 2-state and 3-state \acp{HMM} with moving-average online adaptation is also provided, with closed-form recursions for log-likelihood ratios, detection probabilities, and steady-state approximations, which demonstrate significant robustness improvement over classical sequential authentication schemes.
\end{abstract}


\glsresetall

\maketitle

\section{Introduction}

In the era of ubiquitous wireless connectivity, securing communication systems against sophisticated spoofing attacks has become paramount, particularly in dynamic environments such as \ac{5G}/\ac{6G} networks and \ac{IoT} deployments \cite{GuTIFS2020, BaiJCIN2020, DEVI2025, BagheriGIIS2024,BourjeniJCS2024}. 
In this context, \ac{PLA} emerges as a promising paradigm that verifies transmitter legitimacy by exploiting inherent channel characteristics, thereby circumventing the vulnerabilities of upper-layer cryptographic methods susceptible to computational attacks \cite{Wang2025,XiaoTWC2008, XieCST_2021}. 

Traditional \ac{PLA} schemes often assume channel stationarity and observation independence. 
These assumptions falter in real-world non-stationary channels, where mobility induces temporal-spatial correlations and LoS blockages. 
Machine learning has become increasingly vital in PLA, enabling adaptive feature extraction and pattern recognition in complex, non-stationary channel environments where traditional statistical models fail \cite{SenigagliesiTIEFS_2021, AbdrabouTC_2022, WangCL_2017, FangTC_2019, EZZATIKHATAB2025, MENG2025}.
These limitations are compounded by sophisticated adversarial strategies \cite{YuTIFS2008} exploiting deep learning and generative modeling to mimic legitimate \ac{CSI}, and by adaptive spoofers that reproduce legitimate signals, thereby necessitating robust frameworks that integrate machine learning for feature extraction and adversarial modeling \cite{LiuCOMMST2017,QuiACCESS2020, HoangCST_2024}. 
Moreover, few prior works provide analytical characterizations of detection and false-alarm probabilities under sequential decision rules with temporally correlated observations.

To address these challenges, we propose an \ac{AR-PLA} framework tailored for non-stationary \ac{MIMO} wireless channels \cite{ZhangTWC2020, BagheriJCN2025, boroujeni2025frequency, MiwaWPMC2023}, which inherently offer rich spatial and temporal diversity that can enhance authentication. 
Our approach combines sequential Bayesian decision-making with deep feature extraction via contrastive learning and adversarial modeling using \acp{GAN} to simulate and counter dynamic spoofers \cite{ZhaoTCCN_2025, Xie2024, Xu2022, Shi2021}. 
Unlike conventional methods that overlook channel dynamics \cite{WangINFOCOM2020}, our framework explicitly captures temporal dependencies through 2-state and 3-state \acp{HMM} augmented with moving-average online adaptation, thereby yielding closed-form recursions for log-likelihood ratios, analytical detection probabilities, and steady-state performance guarantees. 
This integrated design yields substantial robustness gains over classical sequential authentication techniques, as confirmed by both analytical characterization and simulation results.

The contributions above are developed as follows: Section \ref{sec:system_model} details the system model and a first strategy for \ac{AR-PLA}, Section \ref{sec:PLA-LoS} presents the \ac{AR-PLA} framework for \ac{LoS} and \ac{NLoS} conditions and Section \ref{sec:perf_alaly} gives a computational complexity analysis and numerical results.

The main contributions of this work are summarized as:

\begin{itemize}
\item \textbf{\ac{AR-PLA} framework:} A novel integration of sequential Bayesian decision-making, deep feature embeddings, and adversarial generative modeling for robust authentication against adaptive spoofers.
\item \textbf{\ac{HMM}-based sequential inference:} Extension of the classical \ac{SPRT} to temporally correlated embeddings via 2-state and 3-state \acp{HMM}, explicitly capturing \ac{LoS}/\ac{NLoS} dynamics, with exponential moving averages for real-time parameter adaptation.
\item \textbf{Non-stationary \ac{MIMO} modeling:} Adoption of a state-space model with temporal evolution and Kronecker-based spatial correlation to realistically capture time-varying \ac{MIMO} channels.
\item \textbf{Analytical characterization:} Derivation of exact recursions for log-posterior ratios, statistical moments of instantaneous \acp{LLR}, the \ac{CDF}/\ac{PDF} of sequential test statistics, and closed-form steady-state approximations.
\end{itemize}

% \vspace{-2ex}
\section{System and Channel Model}
\label{sec:system_model}

We consider a \ac{MIMO} system with a legitimate transmitter (Alice), a legitimate receiver (Bob), and a potential spoofer (Eve).
Alice, Bob, and Eve are equipped with \acp{ULA}, where Alice and Eve use $M_t$ transmit antennas and Bob employs $M_r$ receive antennas.
The non-stationary (time-varying) \ac{MIMO} channel between the transmitter $X \in \{\text{Alice}, \text{Eve}\}$ and receiver (Bob) at time slot $t$ is denoted by
\begin{equation}
\mathbf{H}_X(t) \in \mathbb{C}^{M_r \times M_t}.
\end{equation}

To capture realistic temporal–spatial correlations caused by mobility and propagation dynamics, we adopt a state-space model with Kronecker-based spatial correlation.

%=========================
\subsection{Time--Spatially Correlated Noisy MIMO Channel Model}

We model channel dynamics through a state-space formulation with spatial correlation and noisy observations. Specifically, the \emph{whitened} channel matrix $\mathbf{H}_{w,X}(t) \in \mathbb{C}^{M_r \times M_t}$ evolves as a discrete-time linear state-space process  
\begin{equation}
\mathrm{vec}\big(\mathbf{H}_{w,X}(t)\big) 
= \mathbf{F}_t \, \mathrm{vec}\big(\mathbf{H}_{w,X}(t-1)\big) 
+ \mathbf{w}_X(t),
\label{eq:state_space}
\end{equation}

\begin{figure}[b!]
\centering
\includegraphics[width=\columnwidth]{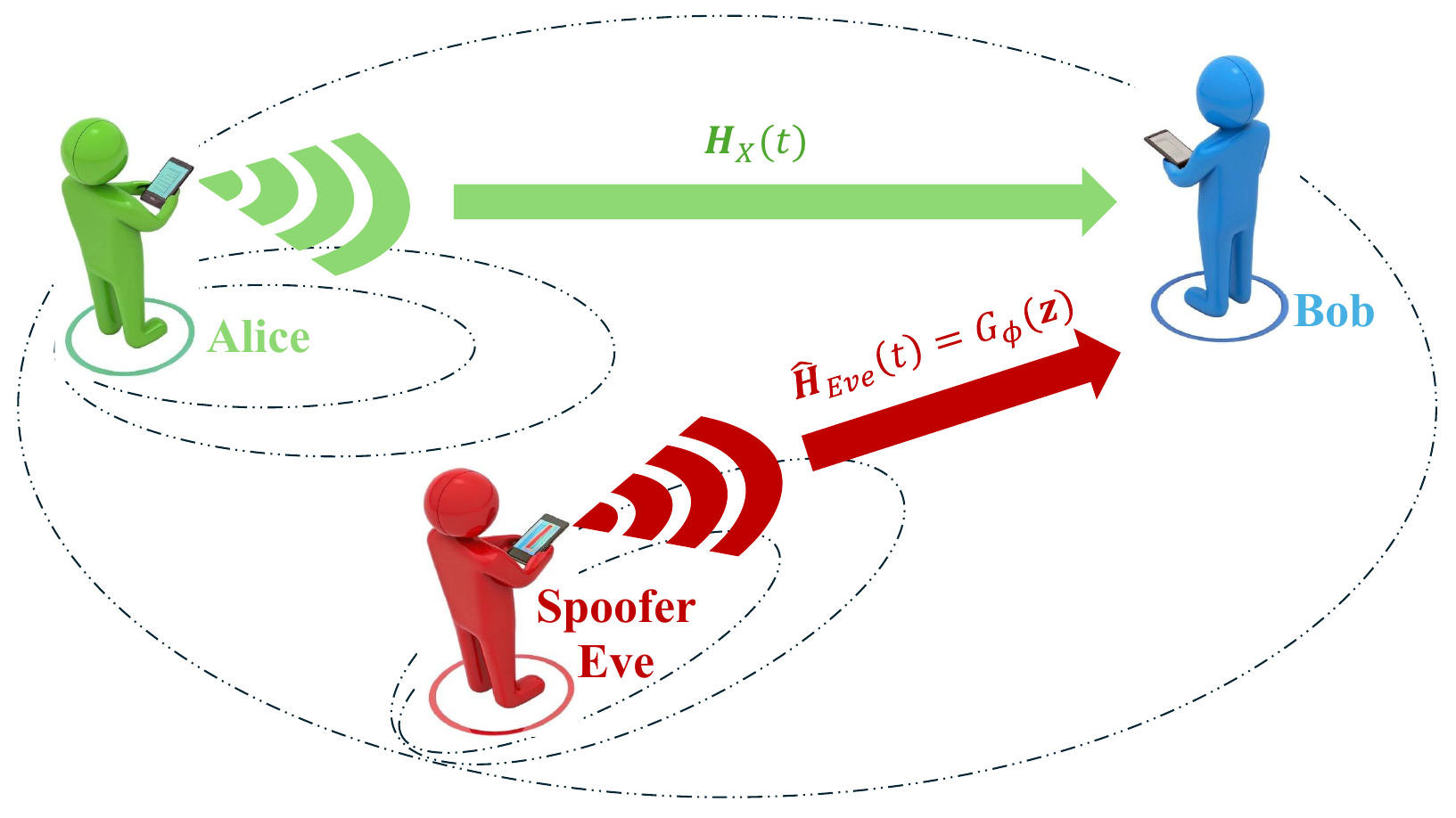}
\vspace{-2ex}
\caption{Illustration of the proposed \ac{AR-PLA} system model, depicting legitimate users Alice and Bob and the adversarial spoofer Eve.}
\label{fig:system_model}
\vspace{-2ex}
\end{figure}

\noindent where $\mathbf{F}_t \in \mathbb{C}^{M_r M_t \times M_r M_t}$ is the (possibly time-varying) state-transition matrix capturing temporal channel dynamics, and $\mathbf{w}_X(t) \sim \mathcal{CN}(\mathbf{0}, \mathbf{Q}_t)$ is the process noise with covariance $\mathbf{Q}_t$, representing unmodeled variations due to mobility and scattering. This formulation captures temporal correlation and non-stationarity in the fading process, with $\mathbf{H}_{w,X}(t)$ denoting the component prior to spatial correlation.

The \emph{spatial correlation} of the transmit and receive arrays is described using the Kronecker model~\cite{Wang_TETCI2018}, leading to the correlated MIMO channel  
\begin{equation}
\mathbf{H}_X(t) = \mathbf{R}_r^{1/2} \, \mathbf{H}_{w,X}(t) \, \mathbf{R}_t^{1/2},
\label{eq:kronecker}
\end{equation}
where $\mathbf{R}_r \in \mathbb{C}^{M_r \times M_r}$ and $\mathbf{R}_t \in \mathbb{C}^{M_t \times M_t}$ are the receive and transmit spatial correlation matrices, respectively. This decomposition separates temporal fading from array correlation yet retains physical interpretability. 

Finally, the receiver’s observation is corrupted by noise, yielding
\begin{equation}
\hat{\mathbf{H}}_X(t) = \mathbf{H}_X(t) + \mathbf{N}_X(t),
\label{eq:noisy_observation}
\end{equation}
where $\mathbf{N}_X(t) \sim \mathcal{CN}(\mathbf{0}, \sigma^2 \mathbf{I})$ denotes additive white Gaussian noise (AWGN), accounting for estimation errors and measurement imperfections.  

This unified model provides a realistic representation of time-varying, spatially correlated MIMO channels under noisy observations, and forms the foundation for the subsequent analysis of physical-layer authentication and secure communications schemes.

%=========================
% \vspace{-2ex}
\subsection{Adversarial Spoofing via Generative Modeling}

To simulate an adaptive and sophisticated spoofer (Eve), we adopt a \ac{GAN} framework, which learns to mimic the distribution of the legitimate transmitter's \ac{CSI}. 
Specifically, Eve's goal is to generate synthetic \ac{CSI} samples that closely resemble those of Alice, thereby deceiving Bob and bypassing physical layer authentication.

Let the generator network be denoted by $G_\phi: \mathbb{R}^d \to \mathbb{C}^{M_r \times M_t}$, parameterized by $\phi$, which maps a latent random vector $\mathbf{z} \sim \mathcal{N}(\mathbf{0}, \mathbf{I}_d)$ sampled from a $d$-dimensional standard complex Gaussian distribution into a synthetic spoofed \ac{CSI} matrix.
Then the aforementioned matrix can be expressed as
\vspace{-1ex}
\begin{equation}
\hat{\mathbf{H}}_{\mathrm{Eve}}(t) = G_\phi(\mathbf{z}).
\vspace{-1ex}
\end{equation}

The discriminator network $D_\psi: \mathbb{C}^{M_r \times M_t} \to [0,1]$, parameterized by $\psi$, is trained to distinguish between real \ac{CSI} samples drawn from the legitimate distribution $p_{\mathrm{Alice}}$ and synthetic samples generated by $G_\phi$. 
The discriminator outputs the probability that an input \ac{CSI} matrix originates from Alice rather than Eve.

The adversarial training procedure follows the classical min-max optimization problem, which is given by
\vspace{-0.5ex}
\begin{align}
\label{eq:gan_loss}
\min_{\phi} \max_{\psi} \quad
&\mathbb{E}_{\mathbf{H} \sim p_{\mathrm{Alice}}} \left[ \log D_\psi(\mathbf{H}) \right] \\[-1ex]
&+ \mathbb{E}_{\mathbf{z} \sim \mathcal{N}(\mathbf{0}, \mathbf{I}_d)} \left[ \log \big(1 - D_\psi(G_\phi(\mathbf{z}))\big) \right],
\nonumber
\vspace{-0.5ex}
\end{align}
where \(D_\psi(\cdot)\) is the discriminator’s output, which can be interpreted as the probability that the input is real (i.e., from Alice), \(G_\phi(\mathbf{z})\) is the generator’s output which is a fake sample generated from a random input \(\mathbf{z}\), \(p_{\mathrm{Alice}}\)  is the true distribution of Alice’s \ac{CSI} and {\(\phi(\mathbf{z})\) is the distribution of latent noise input (e.g., Gaussian).}

At equilibrium, the \ac{GAN} learns to generate spoofed \ac{CSI} samples that closely approximate Alice’s true distribution $p_{\mathrm{Alice}}$, including spatial and temporal correlations. Eve adapts dynamically by retraining $G_\phi$ as channel statistics evolve, achieving realistic spoofing in non-stationary environments. This data-driven approach removes the need for an explicit channel model, while authentication systems typically rely on cumulative \ac{LLR} tests with Gaussian embeddings and stopping rules to detect such attacks.

\vspace{-1ex}
\subsection{Sequential Bayesian Authentication Decision with Temporally Correlated Embeddings}

We consider the problem of authenticating a received signal sequence based on extracted feature embeddings, denoted by $\mathcal{Z}_1^T = \{\mathbf{z}_1, \cdots, \mathbf{z}_t, \cdots, \mathbf{z}_T\}$, where $\mathbf{z}_t \in \mathbb{R}^d$ is the embedding extracted at time step $t$. 
Although the \ac{MIMO} channel is complex-valued, the deep feature extractor outputs real-valued embeddings that both retain discriminative information and enable modeling with a standard multivariate Gaussian \acp{PDF} for the \ac{SPRT}/\ac{HMM} framework.
{\ac{SPRT} is a sequential hypothesis testing method that decides, based on real-valued CSI embeddings collected over time, whether the transmitter is legitimate (Alice) or a spoofer (Eve).}

The encoder maps complex \ac{CSI} into $\mathbb{R}^d$ (e.g., via real/imaginary concatenation or magnitude/phase features), ensuring compatibility with the real-valued Gaussian emission models used in our proposed sequential Bayesian framework. 
The authentication task is formulated as a binary hypothesis test between the legitimate transmitter Alice ($H_0$) and an adversarial spoofer Eve ($H_1$), expressed as
%
% \vspace{-1ex}
\begin{itemize}
\item $H_0$: $\mathbf{z}_t, \forall t$ are generated by Alice,
\item $H_1$: $\mathbf{z}_t, \forall t$ are generated by Eve.
\end{itemize}

% \vspace{-2ex}
\subsection{Sequential Authentication with Independent Observations}

We adopt a sequential decision method that evaluates observations step by step, deciding whether to classify as “Alice,” “Eve,” or wait for more evidence. Under the classical assumption that the embeddings $\{\mathbf{z}_t\}_{t=1}^T$ are \ac{iid} conditioned on the hypothesis, the sequential decision is based on the cumulative \ac{LLR}
\begin{equation}
\Lambda_T = \sum_{t=1}^T \log \frac{p(\mathbf{z}_t \mid H_0)}{p(\mathbf{z}_t \mid H_1)}.
\label{eq:cumulative_llr}
\end{equation}

The class-conditional densities $p(\mathbf{z} \mid H_k)$, $k \in \{0,1\}$, are modeled as multivariate Gaussian distributions, expressed as
\begin{align}
\label{eq:gaussian_pdf}
p(\mathbf{z} \mid H_k) &= \mathcal{N}(\mathbf{z}; \boldsymbol{\mu}_k, \boldsymbol{\Sigma}_k) \\ 
&\hspace{-3ex}= \frac{1}{\sqrt{(2\pi)^d |\boldsymbol{\Sigma}_k|}} \exp \left(-\frac{1}{2} (\mathbf{z} - \boldsymbol{\mu}_k)^\top \boldsymbol{\Sigma}_k^{-1} (\mathbf{z} - \boldsymbol{\mu}_k)\right). \nonumber
\end{align}

Substituting \eqref{eq:gaussian_pdf} into \eqref{eq:cumulative_llr}, the instantaneous \ac{LLR} for the $t$-th sample can be explicitly derived as
\begin{align}
\log \frac{p(\mathbf{z}_t \mid H_0)}{p(\mathbf{z}_t \mid H_1)} &= \frac{1}{2} \Bigg[ \log \frac{|\boldsymbol{\Sigma}_1|}{|\boldsymbol{\Sigma}_0|} - (\mathbf{z}_t - \boldsymbol{\mu}_0)^\top \boldsymbol{\Sigma}_0^{-1} (\mathbf{z}_t - \boldsymbol{\mu}_0) \nonumber \\
&+ (\mathbf{z}_t - \boldsymbol{\mu}_1)^\top \boldsymbol{\Sigma}_1^{-1} (\mathbf{z}_t - \boldsymbol{\mu}_1) \Bigg].
\end{align}

Then, a sequential decision can be made by comparing $\Lambda_T$ against thresholds $(\gamma_0, \gamma_1)$, given by
\begin{equation}
\Lambda_T
\begin{cases}
    \geq \gamma_1, & \text{decide } H_0 \text{ (authenticate as Alice)}, \\
    \leq \gamma_0, & \text{decide } H_1 \text{ (detect spoofing)}, \\
    \text{otherwise}, & \text{continue acquiring observations},
\end{cases}
\label{eq:sequential_decision}
\end{equation}
where $\gamma_0 < 0 < \gamma_1$ are designed to balance the trade-off between false alarms and missed detections. 

This framework corresponds to the classical \ac{SPRT}, known to minimize the expected sample size for given error constraints under \ac{iid} conditions.

% \vspace{-2ex}
\subsection{Extension to Temporally Correlated Embeddings via Hidden Markov Models}

In practical wireless environments, feature embeddings extracted from received signals exhibit temporal correlations caused by channel dynamics, user mobility, and filtering effects. 
To model such temporal dependencies and improve authentication robustness, we extend the framework to incorporate a \ac{HMM} that captures the evolution of the hidden transmitter state.

We model the hidden state sequence $\{S_t\}_{t=1}^T$ as a two-state Markov chain with state space
\begin{equation*}
S_t \in \{0,1\}, \quad 
\begin{cases}
    0: \text{Alice (legitimate)}, \\
    1: \text{Eve (spoofing)}.
\end{cases}
\end{equation*}

The Markov chain is parameterized by the initial state distribution $\boldsymbol{\pi} = [\pi_0, \pi_1]^\top$, where $\pi_k = P(S_1 = k)$, and the state transition matrix
\begin{equation}
\mathbf{A} = \begin{bmatrix}
a_{00} & a_{01} \\
a_{10} & a_{11}
\end{bmatrix}, \quad a_{ij} = P(S_t = j \mid S_{t-1} = i).
\end{equation}

Conditioned on the current hidden state, the observation $\mathbf{z}_t$ is drawn from the emission distribution
\begin{equation}
\mathbf{z}_t \mid S_t = k \sim p(\mathbf{z}_t \mid S_t = k),
\end{equation}
which, for analytical tractability, is often modeled as Gaussian, yielding
\begin{equation}
p(\mathbf{z}_t \mid S_t = k) = \mathcal{N}(\boldsymbol{\mu}_k, \boldsymbol{\Sigma}_k).
\end{equation}

Although deep embeddings may be non-Gaussian, they often cluster elliptically, making Gaussian emissions a justified approximation that enables closed-form inference, simplifies threshold design, and performs well in practice \cite{Wang_TETCI2018}.

% \vspace{-3ex}
\subsubsection{Sequential Bayesian Inference via Forward Algorithm}

The posterior probability of the hidden state given all observations up to time $T$ is
\begin{equation}
\label{eq:SeqBay}
P(S_t = k \mid \mathcal{Z}_1^T) = P(S_t = k \mid \mathbf{z}_1, \ldots, \mathbf{z}_T), \quad k \in \{0,1\}.
\end{equation}

This is efficiently computed via the forward variables $\alpha_t(k)$ defined as
\begin{equation}
\label{eq:alpha_T}
\alpha_t(k) := P(\mathbf{z}_1^T, S_t = k) = p(\mathbf{z}_T \mid S_t = k) \sum_{j=0}^1 \alpha_{t-1}(j) a_{jk},
\end{equation}
with initialization
\begin{equation}
\alpha_1(k) = \pi_k \, p(\mathbf{z}_1 \mid S_1 = k).
\end{equation}

The posterior is obtained by normalization
\begin{equation}
P(S_t = k \mid \mathcal{Z}_1^T) = \frac{\alpha_t(k)}{\sum_{j=0}^1 \alpha_t(j)}.
\end{equation}

We define the log-posterior ratio
\begin{equation}
\label{eq:lambda}
\Lambda_T^{\text{HMM}} := \log \frac{P(S_t = 0 \mid \mathcal{Z}_1^T)}{P(S_t = 1 \mid \mathcal{Z}_1^T)}.
\end{equation}

The sequential decision rule is analogous to \eqref{eq:sequential_decision}, and is given as
\begin{equation}
\Lambda_T^{\text{HMM}} 
\begin{cases}
    \geq \gamma_1, & \text{authenticate as Alice}, \\
    \leq \gamma_0, & \text{detect spoofing}, \\
    \text{otherwise}, & \text{continue observation}.
\end{cases}
\end{equation}

The HMM-based approach leverages temporal correlations and sequential inference to detect intermittent spoofing, model abrupt state changes, and adapt to non-stationary environments, using flexible emission distributions and EM-based parameter estimation for robust physical-layer authentication in dynamic wireless settings.

\subsection{Online Adaptation via Exponential Moving Averages}

To account for the non-stationarity of wireless channels, we adopt an \emph{online adaptation strategy} based on \ac{EMA}. 
At each time step $t$, the estimated mean and covariance of the embedding distribution under hypothesis $H_0$ (Alice) are recursively updated as
\begin{eqnarray}
&\boldsymbol{\mu}_0^{(t)} = \beta \boldsymbol{\mu}_0^{(t-1)} + (1 - \beta) \, \mathbf{z}_t ,&\label{eq:ema_mean} \\
&\boldsymbol{\Sigma}_0^{(t)} = \beta \boldsymbol{\Sigma}_0^{(t-1)}\!+\!(1\!-\!\beta) \, (\mathbf{z}_t\!-\!\boldsymbol{\mu}_0^{(t)})(\mathbf{z}_t\!-\!\boldsymbol{\mu}_0^{(t)})^\top,& \label{eq:ema_cov}
\end{eqnarray}
where $\beta \in [0,1)$ is a \emph{forgetting factor} that controls the adaptation rate. 

Smaller values of $\beta$ result in faster adaptation to new channel conditions, while larger values emphasize historical statistics for smoother updates.
We use the updated mean in the \ac{EMA} covariance computation to ensure the covariance is centered correctly and avoid the small bias introduced by using the previous mean.
This low-complexity recursive update avoids storing historical embeddings and allows real-time adaptation to time-varying embedding distributions induced by dynamic wireless channels, which is beneficial in channels with mobility, scattering changes, or slow-fading. 

\begin{remark}
The same update rule can be independently applied to hypothesis $H_1$ (Eve) if online adaptation of the adversarial distribution is desired, e.g., in the presence of adaptive spoofing attacks.
\end{remark}

% -------------------------
% 2-state EMA-HMM: full derivations & proofs
% -------------------------
\subsection{Analytical Characterization of the 2-State EMA-HMM}
\label{sec:2state_proofs}

We now analyze the statistical behavior of the two-state \ac{HMM} with \ac{EMA} updates. Let $S_t\in\{0,1\}$ denote the hidden state at time $t$ (0: Alice, 1: Eve), with initial distribution $\pi_k=P(S_1=k)$ and transition matrix $A=[a_{ij}]$, where $a_{ij}=P(S_t=j\mid S_{t-1}=i)$. Conditioned on the hidden state, the observed embeddings follow a Gaussian distribution,
\begin{equation}
\mathbf z_t\mid S_t=k \sim \mathcal N(\mu_k,\Sigma_k).
\end{equation}

To evaluate authentication performance, we focus on the forward variables $\alpha_t(k)=P(\mathbf z_1^t,S_t=k)$, which allow us to recursively compute the log-posterior ratio
\begin{equation}
\Lambda_T^{\text{HMM}} := \log\frac{P(S_t=0\mid \mathbf z_1^t)}{P(S_t=1\mid \mathbf z_1^t)}
= \log\frac{\alpha_t(0)}{\alpha_t(1)}.
\end{equation}

\vspace{1ex}
\noindent\textbf{Recursive characterization.}  
The following lemma establishes how $\Lambda_T^{\text{HMM}}$ can be updated in closed form.

\begin{lemma}[Recursive LLR]
\label{lem:recursion}
\quad

Let $\ell_t := \log \tfrac{p_0(\mathbf z_t)}{p_1(\mathbf z_t)}$. Then
\begin{eqnarray}
&\Lambda_T^{\text{HMM}} = \ell_t + \log\dfrac{a_{00}\,r_{t-1} + a_{10}(1-r_{t-1})}{a_{01}\,r_{t-1} + a_{11}(1-r_{t-1})},& \\[1ex]
&r_{t-1} := P(S_{t-1}=0\mid \mathbf z_1^{t-1}) = \sigma(\Lambda_{t-1}^{\text{HMM}}),&\\[1ex]
&\sigma(x)=\dfrac{1}{1+e^{-x}}.&
\end{eqnarray}
\end{lemma}

\begin{proof}
This follows straightforwardly from the forward recursion $\alpha_t(k)=p_k(\mathbf z_t)\sum_j \alpha_{t-1}(j)a_{jk}$ and the definition $r_{t-1} = \alpha_{t-1}(0)/(\alpha_{t-1}(0)+\alpha_{t-1}(1))$.
\end{proof}

\vspace{1ex}
\noindent\textbf{Instantaneous LLR.} The first term $\ell_t$ can be explicitly expressed for Gaussian emissions as
\begin{eqnarray}
\ell_t = \tfrac12 \Big[ \log\tfrac{|\Sigma_1|}{|\Sigma_0|}\!-\!(\mathbf z_t\!-\!\mu_0)^\top \Sigma_0^{-1} (\mathbf z_t\!-\!\mu_0)&&\\
&&\hspace{-13ex} + (\mathbf z_t\!-\!\mu_1)^\top \Sigma_1^{-1} (\mathbf z_t\!-\!\mu_1) \Big].
\nonumber
\end{eqnarray}

Thus, the recursion in Lemma~\ref{lem:recursion} links the instantaneous LLR with the state transition probabilities.

\vspace{1ex}
\noindent\textbf{Distributional moments.}  
To analyze performance, we require the mean and variance of quadratic forms in Gaussian random vectors.

\begin{proposition}[Moments of quadratic forms]
\label{prop:moments_quad}
\quad\\
For $Q_A(\mathbf z)=(\mathbf z-\mu_A)^\top A (\mathbf z-\mu_A)$, $\mathbf z\sim\mathcal N(\mu_B,\Sigma_B)$:
\begin{equation}
\mathbb E[Q_A(\mathbf z)] = \operatorname{tr}(A\Sigma_B) + (\mu_B-\mu_A)^\top A (\mu_B-\mu_A),
% \vspace{-4ex}
\end{equation}
\begin{equation}
\operatorname{Var}[Q_A(\mathbf z)]\!=\!
2\,\operatorname{tr}((A\Sigma_B)^2)\!+\!4(\mu_B\!-\!\mu_A)^\top A \Sigma_B A (\mu_B\!-\!\mu_A).
\end{equation}
Hence, the moments of $\ell_t$ follow directly by linearity.

\end{proposition}

\vspace{1ex}
\noindent\textbf{Equal-covariance case.}  
When $\Sigma_0=\Sigma_1=\Sigma$, the instantaneous LLR simplifies to an affine form
\begin{align}
\ell_t = w^\top \mathbf z_t + \kappa, \quad & \text{where}\quad w=\Sigma^{-1}(\mu_1-\mu_0)\\
&\text{with}\quad \kappa=-\tfrac12(\mu_0^\top\Sigma^{-1}\mu_0-\mu_1^\top\Sigma^{-1}\mu_1).\nonumber
\end{align}

In this case, $\ell_t$ conditioned on $S_t=k$ is Gaussian, $i.e.$
\begin{equation}
\ell_t\mid S_t = k \sim \mathcal N(m_k, v), \; m_k=w^\top\mu_k+\kappa,\; v=w^\top\Sigma w.
\end{equation}

\vspace{1ex}
\noindent\textbf{Recursive distribution of $\Lambda_T^{\text{HMM}}$.}  
The joint effect of temporal correlation and Gaussian emissions can now be summarized as follows.

\begin{theorem}[PDF/CDF recursion, 2-state]
\label{thm:pdf_recursion_2state}
\quad

Let $p_t^{(k)}(y)$ denote the PDF of $\Lambda_T^{\text{HMM}}$ conditioned on $S_t=k$. With $f(x)=\log\tfrac{a_{00}\sigma(x)+a_{10}(1-\sigma(x))}{a_{01}\sigma(x)+a_{11}(1-\sigma(x))}$ and $\phi(\cdot;\mu,v)$ the Gaussian PDF:
\begin{align}
p_t^{(k)}(y) &= \sum_{i=0}^1 a_{ik} \int\limits_{\mathbb R} \phi(y; m_k+f(x),v)\, p_{t-1}^{(i)}(x)\, dx,\\
F_t^{(k)}(\gamma) &= \sum_{i=0}^1 a_{ik} \int\limits_{\mathbb R} \Phi\!\Big(\frac{\gamma\!-\!m_k\!-\!f(x)}{\sqrt v}\Big)\, p_{t-1}^{(i)}(x)\, dx.
\end{align}
\end{theorem}

From these recursions, one can directly compute the operating characteristics of the detector.

\begin{corollary}[Detection / false alarm probabilities]
For thresholds $\gamma_0<0<\gamma_1$:
\begin{equation}
P_\mathrm{FA}(t)=F_t^{(0)}(\gamma_0),\quad P_\mathrm{D}(t)=F_t^{(1)}(\gamma_0).
\end{equation}
\end{corollary}

\vspace{1ex}
\noindent\textbf{Steady-state approximation.}  
To obtain tractable closed-form expressions, we linearize $f(x)\approx c_0+c_1 x$ around the mean $\mu_\Lambda$. This yields an AR(1) approximation
\begin{equation}
\Lambda_T^{\text{HMM}} \!\approx\! c_1 \Lambda_{t-1}^{\text{HMM}} + (\ell_t + c_0), \quad 
\mu_\Lambda \!=\! \tfrac{\bar m+c_0}{1-c_1},\; \sigma_\Lambda^2 \!=\! \tfrac{\operatorname{Var}(\ell_t)}{1-c_1^2},
\end{equation}
so that $\Lambda_T^{\text{HMM}}$ converges to $\mathcal N(\mu_\Lambda,\sigma_\Lambda^2)$ in steady state. Detection and false-alarm probabilities then admit closed-form expressions via the Gaussian CDF $\Phi$.

\vspace{1ex}
\noindent\textbf{Numerical evaluation.}  
For general cases where closed-form expressions are intractable, the integral recursions can be efficiently evaluated using standard quadrature techniques.

%============================================================
\section{Proposed Robust Learning-Aided Authentication Framework with LoS Blockage}
\label{sec:PLA-LoS}

In realistic \ac{MIMO} scenarios, long-term blockage of the \ac{LoS} path between Alice and Bob can occur due to obstacles or environmental changes. 
To account for such events, we extend the previous AR-PLA framework by incorporating a \emph{LoS/NLoS mixed channel model} and a modified sequential authentication algorithm that remains robust under prolonged \ac{LoS} blockage.

The LoS/NLoS \ac{MIMO} channel is modeled as a Rician mixture, given by
\begin{align}
\label{eq:H_full_LosMNLOS}
\mathbf{H}_X(t) &= \sqrt{\tfrac{K_{\text{Rician}}(t)}{K_{\text{Rician}}(t)+1}} \mathbf{H}_X^{\text{LoS}}(t)\\
&\hspace{10ex}+ \sqrt{\tfrac{1}{K_{\text{Rician}}(t)+1}} \underbrace{\mathbf{R}_r^{1/2} \mathbf{H}_{w,X}^{\text{NLoS}}(t) \mathbf{R}_t^{1/2}}_{\text{temporal + spatial correlation}},\nonumber 
\end{align}
where \(\mathbf{H}_X^{\text{LoS}}(t)\) is the deterministic \ac{LoS} component with optional slow phase drift
\begin{equation}
\mathbf{H}_X^{\text{LoS}}(t) = \mathbf{H}_X^{\text{LoS}}(t-1) \, e^{j \phi(t)}, \quad \phi(t) \sim \mathcal{N}(0,\sigma_\phi^2),
\end{equation}
\(\mathbf{H}_{w,X}^{\text{NLoS}}(t)\) evolves as a temporally correlated state-space process given by
\begin{equation}
\mathrm{vec}\big(\mathbf{H}_{w,X}^{\text{NLoS}}(t)\big) = \mathbf{F}_t \, \mathrm{vec}\big(\mathbf{H}_{w,X}^{\text{NLoS}}(t-1)\big) + \mathbf{w}_X(t),
\end{equation}
with $\mathbf{w}_X(t) \sim \mathcal{CN}(\mathbf{0}, \mathbf{Q}_t)$ and spatial correlation is applied via the Kronecker model with \(\mathbf{R}_r, \mathbf{R}_t\).

In addition, the Long-Term \ac{LoS} blockage model can be expressed as
\begin{equation}
% \quad\vspace{-4ex}
\end{equation}
\begin{equation*}
K_{\text{Rician}}(t) =
\begin{cases}
K_0, & \text{LoS available},\\
0, & \text{LoS blocked for } t \in [t_b, t_b + T_{\text{block}}].
\end{cases}
\end{equation*}

During blockage, only the \ac{NLoS} component contributes.

The noisy channel observations under \ac{LoS} blockage can be expressed as
\begin{equation}
\hat{\mathbf{H}}_X(t) = \mathbf{H}_X(t) + \mathbf{N}_X(t), \quad
\mathbf{N}_X(t) \sim \mathcal{CN}(\mathbf{0}, \sigma^2 \mathbf{I}).
\label{eq:3-state_noisy_channel}
\end{equation}

Finally, deep feature embeddings \(\mathbf{z}_t\) are extracted from \(\hat{\mathbf{H}}_X(t)\) using a contrastive learning encoder. 
During \ac{LoS} blockage, the channel statistics change abruptly, requiring adaptive emission parameters, given by
\begin{align}
\mathbf{z}_t \mid S_t &= k \sim \mathcal{N}(\boldsymbol{\mu}_k(t), \boldsymbol{\Sigma}_k(t)), \label{eq:zt}\\
\boldsymbol{\mu}_k(t) &= \beta \boldsymbol{\mu}_k(t-1) + (1-\beta) \mathbf{z}_t, \label{eq:3state_mu}\\
\boldsymbol{\Sigma}_k(t) &= \beta \boldsymbol{\Sigma}_k(t-1) + (1-\beta)(\mathbf{z}_t-\boldsymbol{\mu}_k(t))(\mathbf{z}_t-\boldsymbol{\mu}_k(t))^\top. \label{eq:3state_Sigma}
\end{align}

Using a smaller \(\beta\) during \ac{LoS} blockage helps with faster adaptation to \ac{NLoS}-only statistics.

\subsection{HMM with LoS/NLoS Hidden States}

To leverage NLOS features, we propose the following 3-State \ac{HMM} procedure as follows.
\begin{equation*}
S_t \in \{0,1,2\}, \quad
\begin{cases}
0: \text{Alice, LoS available},\\
1: \text{Alice, LoS blocked (NLoS)},\\
2: \text{Eve (spoofing)}.
\end{cases}
% \vspace{-4ex}
\end{equation*}
\begin{equation}
\end{equation}
% \newpage

If we model the transition matrix $\mathbf{A}_{\text{LoSB}}$ as
\begin{equation}
\mathbf{A}_{\text{LoSB}} =
\begin{bmatrix}
a_{00} & a_{01} & a_{02} \\
a_{10} & a_{11} & a_{12} \\
a_{20} & a_{21} & a_{22} \\
\end{bmatrix}, \quad \sum_{j=0}^2 a_{ij} = 1,
\end{equation}
the forward recursion can be computed as
\begin{equation}
\label{eq:ForRec}
\alpha_t(k) = p(\mathbf{z}_t \mid S_t = k) \sum_{j=0}^2 \alpha_{t-1}(j) a_{jk}, \quad k=0,1,2,
\end{equation}
where the posterior probabilities are given by
\begin{equation}
P(S_t = k \mid \mathcal{Z}_1^T) = \frac{\alpha_t(k)}{\sum_{j=0}^2 \alpha_t(j)}.
\end{equation}

Accordingly, the log-posterior ratio for authentication is defined as
\begin{equation}
\Lambda_T^{\text{LoSB-HMM}} := \log \frac{P(S_t \in \{0,1\} \mid \mathcal{Z}_1^T)}{P(S_t = 2 \mid \mathcal{Z}_1^T)}.
\label{eq:3state_Log-posterior_ratio}
\end{equation}

Finally, the sequential decision rule under LoS blockage can be expressed as
\begin{equation*}
\Lambda_T^{\text{LoSB-HMM}}
\begin{cases}
\geq \gamma_1, & \text{authenticate as Alice},\\
\leq \gamma_0, & \text{detect spoofing},\\
\text{otherwise}, & \text{continue observation}.
\end{cases}
% \vspace{-5ex}
\end{equation*}
\begin{equation}
\end{equation}

\begin{remark}
NLoS temporal and spatial correlations capture realistic channel dynamics, even during LoS blockage, while LoS is modeled as deterministic with optional slow phase drift. Online adaptation of emission distributions enhances robustness to sudden LoS loss, and the sequential Bayesian HMM naturally integrates LoS and NLoS states of Alice as legitimate, enabling reliable distinction from Eve.
\end{remark}

% -------------------------
% 3-state LoS-HMM: full derivations & proofs
% -------------------------
\subsection{Analytical characterization of the 3-state LoS-HMM}
\label{sec:3state_proofs}

\paragraph*{Model and notation.}  
Consider the hidden state $S_t\in\{0,1,2\}$, with $0$: Alice (LoS), $1$: Alice (NLoS), and $2$: Eve. Observations follow $\mathbf z_t\mid S_t=k \sim \mathcal N(\mu_k,\Sigma_k)$. Let $A=[a_{ij}]_{i,j=0}^2$ be the transition matrix and $\alpha_t(k)=P(\mathbf z_1^t,S_t=k)$ the forward variables. The posterior mass of Alice’s states is
\begin{equation}
P_A(t) := P(S_t\in\{0,1\}\mid\mathbf z_1^t)=
\frac{\alpha_t(0)+\alpha_t(1)}{\sum_{j=0}^2 \alpha_t(j)}.
\end{equation}
The log-posterior ratio (LoSB-HMM statistic) is
\begin{equation}
\Lambda_T^{\mathrm{LoSB-HMM}}
=\log\frac{P(S_t\in\{0,1\}\mid\mathbf z_1^t)}{P(S_t=2\mid\mathbf z_1^t)}
=\log\frac{\alpha_t(0)+\alpha_t(1)}{\alpha_t(2)}.
\end{equation}

\paragraph*{Exact representation.}
Define $u_{t-1}(j):=P(S_{t-1}=j\mid\mathbf z_1^{t-1})$ and the transition-weighted priors
\begin{equation}
T_k(u_{t-1}) := \sum_{j=0}^2 a_{jk}\,u_{t-1}(j),\quad k=0,1,2.
\end{equation}
For pairwise LLRs $\ell_{k2}(\mathbf z_t)=\log\frac{p_k(\mathbf z_t)}{p_2(\mathbf z_t)}$, $k=0,1$, we obtain:

\begin{lemma}[3-state representation]
\label{lem:3state_repr}
\begin{equation}
\Lambda_T^{\mathrm{LoSB-HMM}}
= \log\!\big( T_0 e^{\ell_{02}(\mathbf z_t)} + T_1 e^{\ell_{12}(\mathbf z_t)} \big)
- \log T_2.
\end{equation}
\end{lemma}
% \newpage

\begin{proof}
From the forward recursion $\alpha_t(k)=p_k(\mathbf z_t)\sum_j\alpha_{t-1}(j)a_{jk}$, normalizing by $\sum_i\alpha_{t-1}(i)$ yields
\begin{equation}
\frac{\alpha_t(0)+\alpha_t(1)}{\alpha_t(2)}=\frac{p_0(\mathbf z_t)T_0+p_1(\mathbf z_t)T_1}{p_2(\mathbf z_t)T_2}.
\end{equation}

Taking logs gives the claim.
\end{proof}

\paragraph*{Equal-covariance case.}
If $\Sigma_k=\Sigma$ for all $k$, then each $\ell_{k2}(\mathbf z)$ is affine, such that
\begin{align}
\ell_{k2}(\mathbf z) = w_{k2}^\top \mathbf z + \kappa_{k2},\;\; &\text{where } w_{k2} = \Sigma^{-1}(\mu_2-\mu_k),\\
&\hspace{-11ex}\text{with } \kappa_{k2}=-\tfrac12(\mu_k^\top\Sigma^{-1}\mu_k-\mu_2^\top\Sigma^{-1}\mu_2).
\end{align}

Thus
\begin{equation}
\mathbf L_t :=
\begin{bmatrix}\ell_{02}(\mathbf z_t)\\ \ell_{12}(\mathbf z_t)\end{bmatrix}
\;\big|\; S_t=s \;\sim\; \mathcal N(\mathbf m_{|s},\mathbf V),
\end{equation}
with mean $\mathbf m_{|s}$ and covariance $\mathbf V$ given by
\begin{eqnarray}
&\mathbf m_{|s}
\begin{bmatrix}w_{02}^\top\mu_s+\kappa_{02}\\
[2pt] w_{12}^\top\mu_s+\kappa_{12}\end{bmatrix},&\\
&\mathbf V=
\begin{bmatrix}w_{02}^\top\Sigma w_{02} & w_{02}^\top\Sigma w_{12}\\[2pt]
w_{12}^\top\Sigma w_{02} & w_{12}^\top\Sigma w_{12}
\end{bmatrix}.&
\label{eq:bivar_normal}
\end{eqnarray}

\begin{theorem}[CDF of $\Lambda_T^{\text{LoSB-HMM}}$ under equal covariances]
\label{thm:3state_cdf}
Conditioned on $u_{t-1}$ and $S_t=s$, the CDF is
\begin{align}
F_{t\mid u}^{(s)}(\gamma)
&=\iint_{\; \mathcal R(\gamma)} \varphi_2(\ell_{02},\ell_{12};\mathbf m_{|s},\mathbf V)\;
d\ell_{02}\,d\ell_{12}, \\
\mathcal R(\gamma)&=\Big\{(\ell_{02},\ell_{12}): T_0 e^{\ell_{02}}+T_1 e^{\ell_{12}}\le e^\gamma T_2\Big\},
\end{align}
where $\varphi_2$ is the bivariate normal PDF. The PDF follows by differentiation.
\end{theorem}

\paragraph*{Approximate characterizations.}
Since the exact distribution involves bivariate integration, we use approximations:  

1. \emph{Delta-method.}

With $g(x,y)=\log(T_0 e^x+T_1 e^y)-\log T_2$,
\begin{equation}
\Lambda_T^{\text{LoSB-HMM}}=g(\ell_{02},\ell_{12}).
\end{equation}

For $(\ell_{02},\ell_{12})\sim\mathcal N(\mathbf m,\mathbf V)$,
\begin{align}
\mathbb E[\Lambda_T^{\text{LoSB-HMM}}] &\approx g(\mathbf m)+\tfrac12\operatorname{tr}(H_g(\mathbf m)\mathbf V), \\
\operatorname{Var}[\Lambda_T^{\text{LoSB-HMM}}] &\approx \nabla g(\mathbf m)^\top \mathbf V \nabla g(\mathbf m),
\end{align}
with explicit yields
\begin{equation}
\nabla g=\Big(\tfrac{T_0 e^x}{S},\;\tfrac{T_1 e^y}{S}\Big)^\top,\,
H_g=\frac{T_0T_1 e^{x+y}}{S^2}
\begin{bmatrix}1&-1\\-1&1\end{bmatrix},
\end{equation}
evaluated at $\mathbf m$, and such that the Gaussian approximation then yields conditional CDFs.

2. \emph{Log-sum-exp / Laplace.}

Writing
\begin{align}
&\log(T_0 e^X+T_1 e^Y)= \max\{X+\log T_0,Y+\log T_1\}\\
&\hspace{23ex}+\log(1+e^{-|X-Y+\log(T_0/T_1)|}),\nonumber
\end{align}
reveals that $\Lambda_T$ is well-approximated by the dominant affine Gaussian term, with corrections obtained via Laplace or saddlepoint expansions, especially accurate in the tails.

\begin{algorithm}[H]
\caption{LoS/NLoS-Aware Adversarially Robust Learning-Aided Physical Layer Authentication (AR-PLA)}
\label{alg:ar-pla-los-bob}
\begin{algorithmic}[1]
\Require Pre-trained encoder, HMM parameters, thresholds $\gamma_0,\gamma_1$, smoothing factor $\beta$, initial $K_{\text{Rician}}$
\State \textbf{Initialize} emission statistics and HMM forward variables
\For{each time slot $t = 1, 2, \ldots$}
\State Receive noisy CSI observation $\hat{\mathbf{H}}_X(t)$ as in \eqref{eq:3-state_noisy_channel}
\State Extract feature embedding $\mathbf z_t$ using the contrastive encoder
\State Update emission means and covariances with the smoothing rule as in \eqref{eq:3state_mu} and \eqref{eq:3state_Sigma}
\State Run HMM forward recursion using \eqref{eq:ForRec}
\State Compute log-posterior ratio $\Lambda_T^{\text{LoSB-HMM}}$ as in \eqref{eq:3state_Log-posterior_ratio}
\State \textbf{Sequential decision:}
\If{$\Lambda_T^{\text{LoSB-HMM}} \geq \gamma_1$}
    \State Authenticate as Alice
\ElsIf{$\Lambda_T^{\text{LoSB-HMM}} \leq \gamma_0$}
    \State Detect spoofing (Eve)
\Else
    \State Continue observation
\EndIf
\State Adapt emission statistics with smaller $\beta$ if long-term LoS blockage is detected
\EndFor
\Ensure Authentication decision at each time step with adaptive model update
\end{algorithmic}
\end{algorithm}

\paragraph*{Detection performance.}
Finally, with hypotheses $H_0:S_t\in\{0,1\}$ (Alice) and $H_1:S_t=2$ (Eve), threshold $\gamma_0$ yields
\begin{align}
P_{\mathrm{FA}}(t)
&=\Pr[\Lambda_T^{\text{LoSB-HMM}}\le\gamma_0\mid H_0] \nonumber \\
&=\frac{\rho_t(0)F_t^{(0)}(\gamma_0)+\rho_t(1)F_t^{(1)}(\gamma_0)}{\rho_t(0)+\rho_t(1)}, \\
P_{\mathrm{D}}(t)
&=\Pr[\Lambda_T^{\text{LoSB-HMM}}\le\gamma_0\mid H_1]=F_t^{(2)}(\gamma_0),
\end{align}
where $\rho_t(s)=P(S_t=s)$ and $F_t^{(s)}$ is the conditional CDF from Theorem~\ref{thm:3state_cdf} or its approximations.

\paragraph{Recursive propagation and practical considerations.}  
Exact computation of the distribution of $\Lambda_T^{\mathrm{LoSB-HMM}}$ requires tracking the joint posterior vector $u_{t-1} = [P(S_{t-1}=0), P(S_{t-1}=1), P(S_{t-1}=2)]^\top$ and the law of the previous LLR summary. In practice, this can be approximated as follows:  

\begin{enumerate}
\item Maintain $u_{t-1}$ (or a sufficient summary, e.g., the total Alice probability) via the forward algorithm.  
\item For each candidate $u_{t-1}$ (grid point or particle), compute the transition-weighted priors $T_k(u_{t-1})$ and evaluate the conditional bivariate Gaussian integral numerically.  
\item Aggregate over the distribution of $u_{t-1}$ (via exact summation, grid averaging, or particle filtering) to obtain the marginal distribution of $\Lambda_T$.
\end{enumerate}

\begin{table}[b!]
\centering
\caption{Online computational complexity per time step for AR-PLA.}
\begin{tabular}{l l}
\toprule
\textbf{Module} & \textbf{Complexity / Function} \\
\hline
GAN generator & $O(L_g d M_r M_t)$; feature embeddings \\
Encoder & $O(L_e d M_r M_t)$; latent feature extraction \\
Gaussian LLR & $O(N_s d^2)$; log-likelihood computation \\
EMA update & $O(N_s d^2)$; statistics adaptation \\
HMM forward & $O(N_s^2)$; hidden state recursion \\
Log-posterior ratio & $O(N_s)$; authentication decision \\
\hline
\end{tabular}
\label{tab:complexity_onecol}
% \vspace{-5ex}
\end{table}

When all covariances are equal, the integrals reduce to two-dimensional quadratures, and the delta-method Gaussian approximation provides simple closed-form estimates of $F_t^{(s)}(\cdot)$ suitable for threshold design.

\section{Performance Analysis}
\label{sec:perf_alaly}

\subsection{Computational Complexity Analysis}
\label{sec:complexity}

The computational complexities of the proposed method are summarized in Table \ref{tab:complexity_onecol}. As shown, the dominant operations are associated with the \ac{GAN} generator and feature extraction, scaling linearly with the number of layers, latent dimensions, and \ac{CSI} size. Gaussian \ac{LLR} updates scale quadratically with the latent dimension but remain lightweight for typical settings $(d<50)$. The \ac{HMM} forward recursion and decision computations involve only a small number of hidden states and therefore incur negligible cost. Overall, all components are computationally efficient and suitable for real-time implementation, even in moderate-sized \ac{MIMO} systems, ensuring that the proposed authentication framework is practical for online deployment.

\subsection{Simulation Results}
\vspace{-3ex}
\begin{figure}[H]
\centering
\setlength{\abovecaptionskip}{0pt}  % reduce space above captions
\setlength{\belowcaptionskip}{5pt}  % reduce space below captions
\includegraphics[width=0.95\columnwidth]{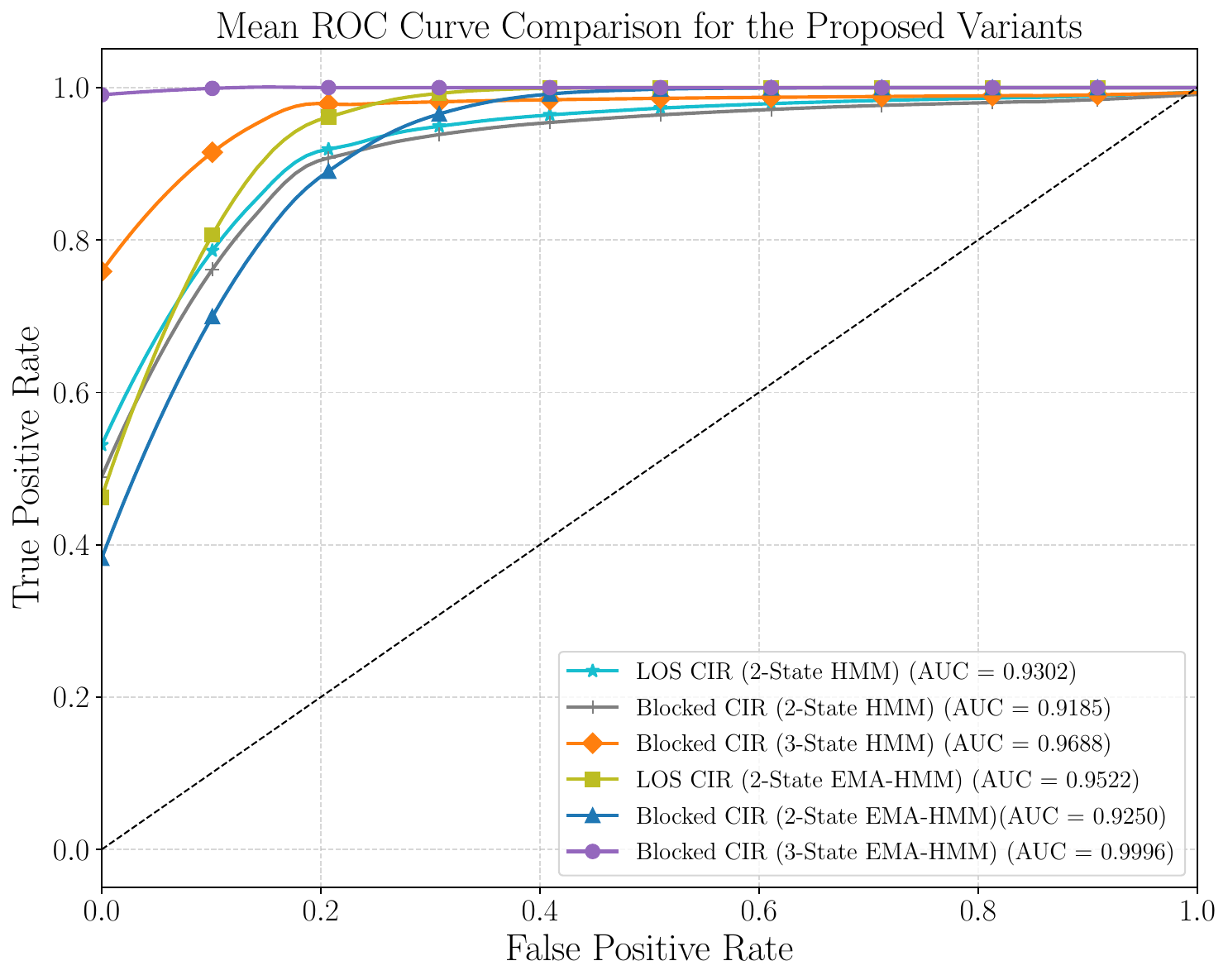}
\caption{ROC performance of the proposed \ac{AR-PLA} scheme under different HMM and blockage settings.}
\label{fig:roc_curves}
% \vspace{-1ex}
\includegraphics[width=\columnwidth]{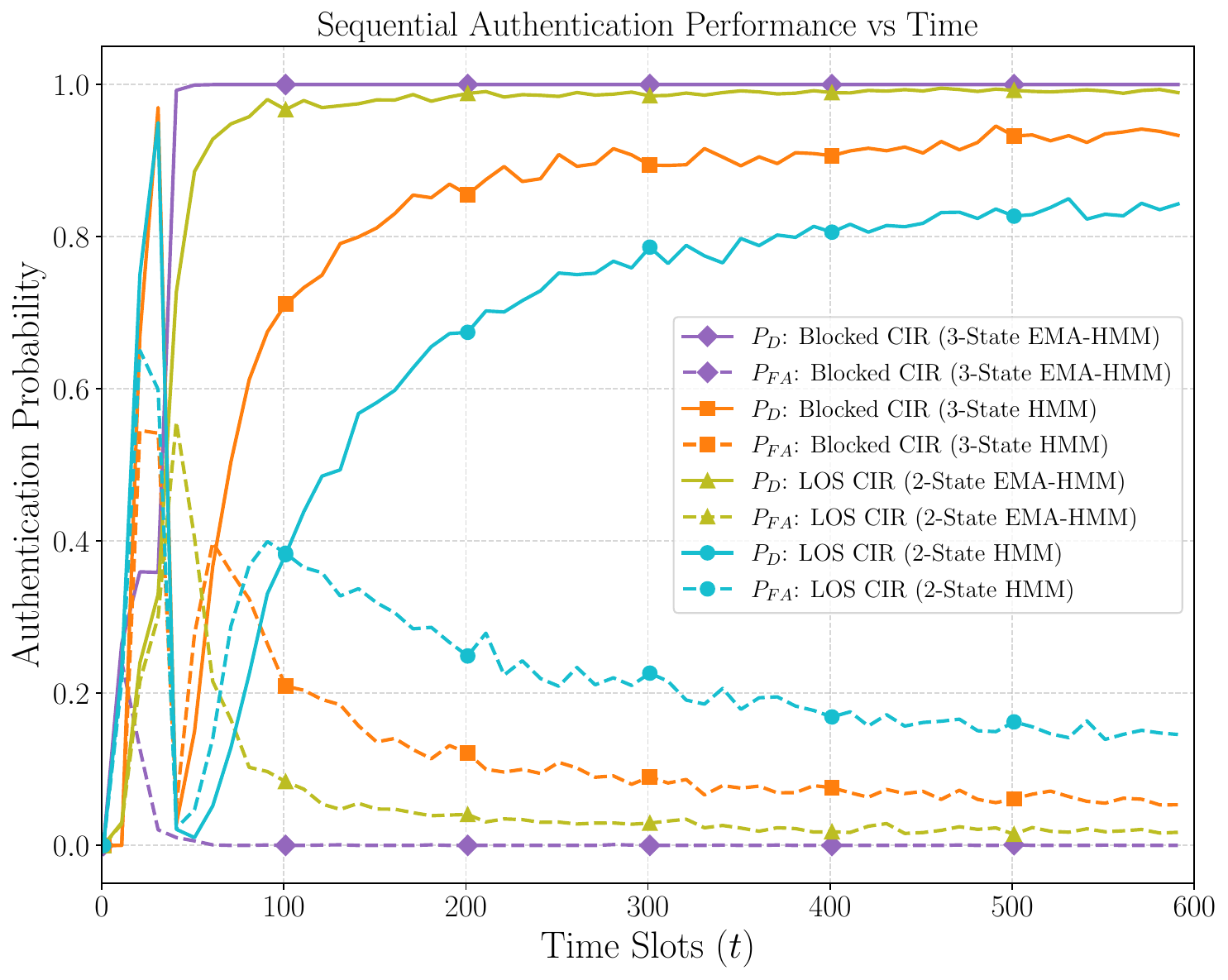}
\vspace{-1ex}
\caption{Authentication probability and convergence behavior of the proposed \ac{AR-PLA} scheme.}
\label{fig:Probability}
% \vspace{-1ex}
\includegraphics[width=\columnwidth]{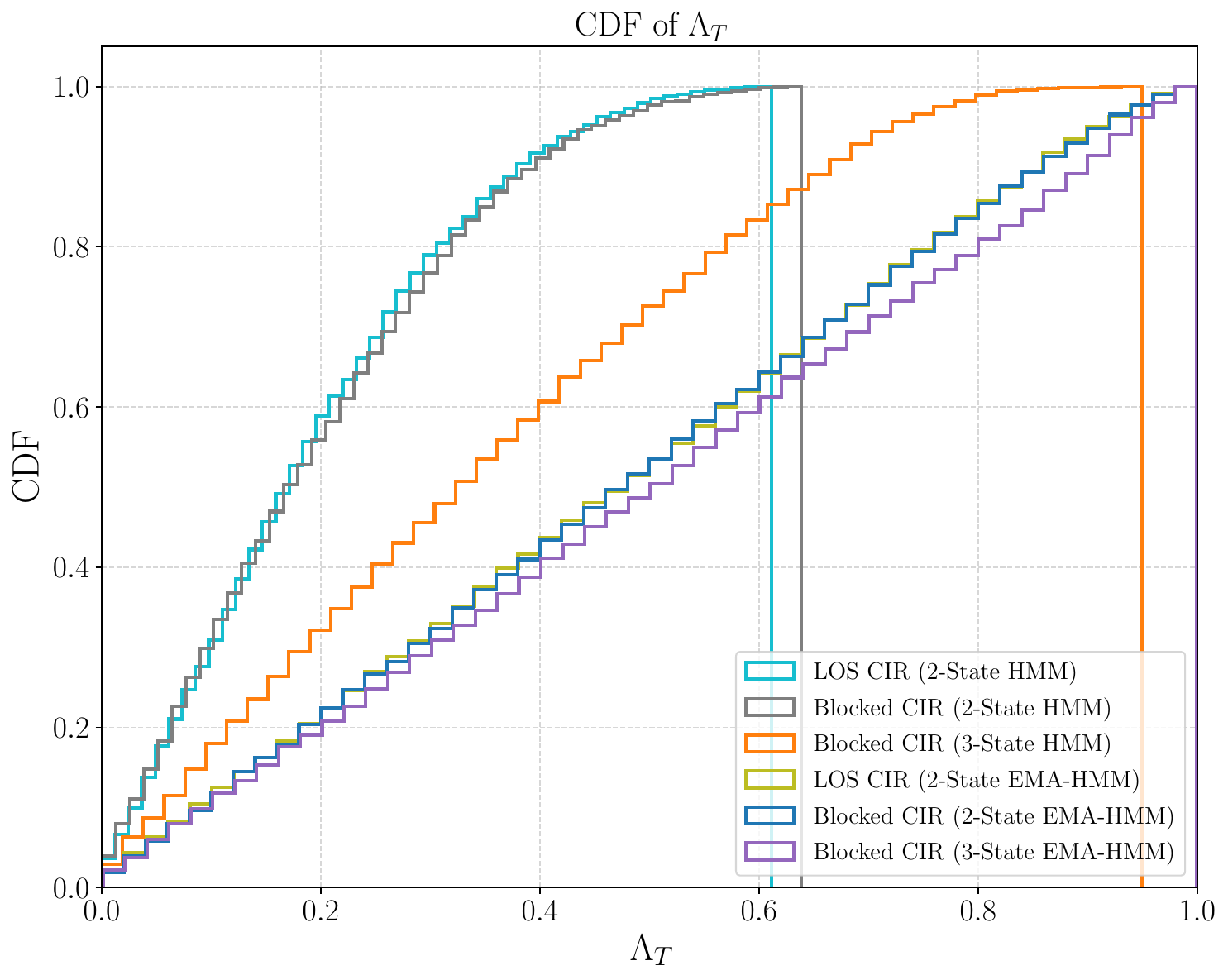}
\caption{CDF of posterior authentication probability for the proposed \ac{AR-PLA} scheme.}
\label{fig:CDF}
\end{figure}

We evaluate the proposed \ac{AR-PLA} framework under the simulation parameters summarized in Table~\ref{tab:sim_params}.

This structured setup provides a consistent baseline for evaluating the impact of different HMM variants, LoS/NLoS blockage modeling, and online adaptation via exponential moving averages (\ac{EMA}) on authentication performance.

Figure~\ref{fig:roc_curves} compares six scenarios of our proposed \ac{PLA} framework. The first three include: (i) 2-state \ac{HMM} in LoS channels, (ii) 2-state \ac{HMM} with potential blockage, and (iii) 3-state \ac{HMM} under blockage. AUC values remain above 0.91 even with a GAN spoofer, dropping from 0.93 to 0.91 for the 2-state HMM under blockage and rising to ~0.97 for the 3-state HMM.

The next three scenarios apply EMA to the same variants, substantially improving performance. EMA enables near-perfect discrimination (\ac{AUC} $\approx 1$) even under combined blockage and GAN spoofing, highlighting the robustness of both the 3-state HMM and EMA smoothing.

As shown in Figure~\ref{fig:Probability}, the 3-state \ac{HMM} improves detection probability, reduces false alarms, and decreases the training time needed for reliable authentication. EMA-based adaptation further enhances performance by boosting detection probability, lowering false alarms, and accelerating convergence. Specifically, the 3-state \ac{HMM} with blockage and EMA achieves an AUC near unity with fewer consecutive decisions. Deploying EMA increases detection probability to nearly 1, reduces false alarms to about 0, and substantially shortens decision time, demonstrating that the combination of 3-state \ac{HMM} and EMA improves both accuracy and efficiency.

Finally, Figure~\ref{fig:CDF} shows the CDF of posterior authentication probabilities. 
The 3-state \ac{HMM} shifts the distribution rightward, indicating reduced authentication errors. 
EMA updates further accentuate this effect, confirming their benefit in fast-varying and non-stationary channel conditions.

\begin{table}[t!]
\centering
\caption{Simulation Parameters for AR-PLA Evaluation}
\label{tab:sim_params}
\begin{tabular}{ll}
\toprule
\textbf{Parameter} & \textbf{Value} \\
\midrule
Signal-to-Noise Ratio (SNR) & 5 dB \\
Transmit/Receive antennas ($M_t=M_r$) & 4 \\
Temporal correlation ($\rho_t$) & 0.7 \\
Embedding dimension ($d$) & 16 \\
Initial distribution (2-state HMM) & $[0.7,\,0.3]^\top$ \\
Initial distribution (3-state HMM) & $[0.45,\,0.45,\,0.1]^\top$ \\
Rician factor ($K_0$) & 10 \\
\bottomrule
\end{tabular}
\end{table}

\section*{Conclusion}

We proposed an \ac{AR-PLA} framework for non-stationary \ac{MIMO} channels, integrating sequential Bayesian decision-making, contrastive feature extraction, and generative adversarial modeling to counter adaptive spoofers. 
Unlike conventional methods assuming channel stationarity or independent observations, our approach accounts for temporal-spatial correlations, \ac{LoS} blockages, and evolving attack strategies, enhancing resilience in realistic environments.

Analytical characterization using 2-state and 3-state \acp{HMM} with exponential moving average adaptation demonstrates significant robustness gains. 
Derived closed-form recursions for log-likelihood ratios and detection probabilities confirm faster convergence and reduced computational complexity compared to classical sequential authentication.

Overall, the framework effectively combines statistical rigor with learning-based adaptability, achieving high detection accuracy and low false alarms, supporting secure and trustworthy \ac{5G}/\ac{6G} and \ac{IoT} wireless systems.

\vspace{-0.5em}
\section*{Acknowledgment}
This work was funded by the German Federal Ministry of Research, Technology and Space (grant 16KISK231 and grant 16KIS2343) and based on the budget passed by the Saxon State Parliament.

% \newpage
\bibliographystyle{ieeetr}

\begin{thebibliography}{10}

\bibitem{GuTIFS2020}
Z.~Gu, H.~Chen, P.~Xu, Y.~Li, and B.~Vucetic, ``{P}hysical {L}ayer {A}uthentication for {N}on-{C}oherent {M}assive {SIMO}-{E}nabled {I}ndustrial {IoT} {C}ommunications,'' {\em IEEE Transactions on Information Forensics and Security}, vol.~15, pp.~3722--3733, 2020.

\bibitem{BaiJCIN2020}
L.~Bai, L.~Zhu, J.~Liu, J.~Choi, and W.~Zhang, ``{P}hysical {L}ayer {A}uthentication in {W}ireless {C}ommunication {N}etworks: {A} {S}urvey,'' {\em Journal of Communications and Information Networks}, vol.~5, no.~3, pp.~237--264, 2020.

\bibitem{DEVI2025}
P.~Devi, M.~R. Bharti, and D.~Gautam, ``{A} {S}urvey on {P}hysical {L}ayer {S}ecurity for {5G/6G} {C}ommunications over {D}ifferent {F}ading {C}hannels: {A}pproaches, {C}hallenges, and {F}uture {D}irections,'' vol.~53, p.~Art. no. 100891, 2025.

\bibitem{BagheriGIIS2024}
G.~Bagheri, A.~K. Boroujeni, and S.~Köpsell, ``{M}achine {L}earning-{B}ased {V}ector {Q}uantization for {S}ecret {K}ey {G}eneration in {P}hysical {L}ayer {S}ecurity,'' in {\em Proc. IEEE Global Information Infrastructure and Networking Symp. (GIIS)}, pp.~1--7, 2024.

\bibitem{BourjeniJCS2024}
A.~K. Boroujeni, G.~Bagheri, and S.~Köpsell, ``{E}nhancing {F}requency {H}opping {S}ecurity in {ISAC} {S}ystems: {A} {P}hysical {L}ayer {S}ecurity {A}pproach,'' in {\em Proc. IEEE Int. Symp. Joint Commun. Sensing (JC\&S)}, pp.~1--6, 2024.

\bibitem{Wang2025}
X.~Wang, P.~Hao, and L.~Hanzo, ``{P}hysical-{L}ayer {A}uthentication for {W}ireless {S}ecurity {E}nhancement: {C}urrent {C}hallenges and {F}uture {D}evelopments,'' {\em IEEE Communications Magazine}, vol.~54, no.~6, pp.~152--158, 2016.

\bibitem{XiaoTWC2008}
L.~Xiao, L.~J. Greenstein, N.~B. Mandayam, and W.~Trappe, ``{U}sing the {P}hysical {L}ayer for {W}ireless {A}uthentication in {T}ime-{V}ariant {C}hannels,'' {\em IEEE Transactions on Wireless Communications}, vol.~7, no.~7, pp.~2571--2579, 2008.

\bibitem{XieCST_2021}
N.~Xie, Z.~Li, and H.~Tan, ``{A} {S}urvey of {P}hysical-{L}ayer {A}uthentication in {W}ireless {C}ommunications,'' {\em IEEE Communications Surveys \& Tutorials}, vol.~23, no.~1, pp.~282--310, 2021.

\bibitem{SenigagliesiTIEFS_2021}
L.~Senigagliesi, M.~Baldi, and E.~Gambi, ``{C}omparison of {S}tatistical and {M}achine {L}earning {T}echniques for {P}hysical {L}ayer {A}uthentication,'' {\em IEEE Transactions on Information Forensics and Security}, vol.~16, pp.~1506--1521, 2021.

\bibitem{AbdrabouTC_2022}
M.~Abdrabou and T.~A. Gulliver, ``{A}daptive {P}hysical {L}ayer {A}uthentication {U}sing {M}achine {L}earning with {A}ntenna {D}iversity,'' {\em IEEE Transactions on Communications}, vol.~70, no.~10, pp.~6604--6614, 2022.

\bibitem{WangCL_2017}
N.~Wang, T.~Jiang, S.~Lv, and L.~Xiao, ``{P}hysical-{L}ayer {A}uthentication {B}ased on {E}xtreme {L}earning {M}achine,'' {\em IEEE Communications Letters}, vol.~21, no.~7, pp.~1557--1560, 2017.

\bibitem{FangTC_2019}
H.~Fang, X.~Wang, and L.~Hanzo, ``{L}earning-{A}ided {P}hysical {L}ayer {A}uthentication as an {I}ntelligent {P}rocess,'' {\em IEEE Transactions on Communications}, vol.~67, no.~3, pp.~2260--2273, 2019.

\bibitem{EZZATIKHATAB2025}
Z.~Ezzati~Khatab, A.~Mohammadi, V.~Pourahmadi, and A.~Kuhestani, ``{A} {M}achine {L}earning-{B}ased {P}hysical {L}ayer {A}uthentication with {P}hase {I}mpairments,'' vol.~68, p.~Art. no. 102545, 2025.

\bibitem{MENG2025}
R.~Meng, B.~Xu, X.~Xu, M.~Sun, B.~Wang, S.~Han, S.~Lv, and P.~Zhang, ``{A} {S}urvey of {M}achine {L}earning-{B}ased {P}hysical-{L}ayer {A}uthentication in {W}ireless {C}ommunications,'' vol.~235, p.~Art. no. 104085, 2025.

\bibitem{YuTIFS2008}
P.~L. Yu, J.~S. Baras, and B.~M. Sadler, ``{P}hysical-{L}ayer {A}uthentication,'' {\em IEEE Transactions on Information Forensics and Security}, vol.~3, no.~1, pp.~38--51, 2008.

\bibitem{LiuCOMMST2017}
Y.~Liu, H.-H. Chen, and L.~Wang, ``{P}hysical {L}ayer {S}ecurity for {N}ext {G}eneration {W}ireless {N}etworks: {T}heories, {T}echnologies, and {C}hallenges,'' {\em IEEE Communications Surveys \& Tutorials}, vol.~19, no.~1, pp.~347--376, 2017.

\bibitem{QuiACCESS2020}
X.~Qiu, J.~Dai, and M.~Hayes, ``{A} {L}earning {A}pproach for {P}hysical {L}ayer {A}uthentication {U}sing {A}daptive {N}eural {N}etwork,'' {\em IEEE Access}, vol.~8, pp.~26139--26149, 2020.

\bibitem{HoangCST_2024}
T.~M. Hoang, A.~Vahid, H.~D. Tuan, and L.~Hanzo, ``{P}hysical {L}ayer {A}uthentication and {S}ecurity {D}esign in the {M}achine {L}earning {E}ra,'' {\em IEEE Communications Surveys \& Tutorials}, vol.~26, no.~3, pp.~1830--1860, 2024.

\bibitem{ZhangTWC2020}
P.~Zhang, T.~Taleb, X.~Jiang, and B.~Wu, ``{P}hysical {L}ayer {A}uthentication for {M}assive {MIMO} {S}ystems with {H}ardware {I}mpairments,'' {\em IEEE Transactions on Wireless Communications}, vol.~19, no.~3, pp.~1563--1576, 2020.

\bibitem{BagheriJCN2025}
G.~Bagheri, P.~Walther, M.~Braunig, A.~K. Boroujeni, and S.~Köpsell, ``{F}eature {E}xtraction for {C}hannel {R}eciprocity {B}ased {S}ecret {K}ey {G}eneration {M}ethods,'' {\em Journal of Communications and Networks}, vol.~27, no.~3, pp.~147--165, 2025.

\bibitem{boroujeni2025frequency}
A.~K. Boroujeni, G.~T.~F. de~Abreu, S.~Köpsell, G.~Bagheri, K.~R.~R. Ranasinghe, and R.~F. Schaefer, ``{F}requency {H}opping {W}aveform {D}esign for {S}ecure {I}ntegrated {S}ensing and {C}ommunications,'' {\em arXiv preprint arXiv:2504.10052}, 2025.

\bibitem{MiwaWPMC2023}
K.~Miwa, K.~Ando, G.~T.~F. de~Abreu, and K.~Ishibashi, ``{P}hysical-{L}ayer {A}uthentication {O}ver {S}patially {C}orrelated {MIMO} {C}hannels,'' in {\em Proc. Int. Symp. Wireless Pers. Multimedia Commun. (WPMC)}, pp.~1--6, 2023.

\bibitem{ZhaoTCCN_2025}
C.~Zhao, H.~Du, D.~Niyato, J.~Kang, Z.~Xiong, D.~I. Kim, X.~Shen, and K.~B. Letaief, ``{G}enerative {AI} for {S}ecure {P}hysical {L}ayer {C}ommunications: {A} {S}urvey,'' {\em IEEE Transactions on Cognitive Communications and Networking}, vol.~11, no.~1, pp.~3--26, 2025.

\bibitem{Xie2024}
W.~Xie, H.~Wang, Z.~Feng, and C.~Ma, ``{P}hysical {L}ayer {A}uthentication {B}ased on {CNN}-{GAN},'' in {\em Proc. IEEE Int. Conf. Wireless Commun. Signal Process. (WCSP)}, pp.~572--577, 2024.

\bibitem{Xu2022}
T.~Xu and Z.~Wei, ``{W}aveform {D}efence {A}gainst {D}eep {L}earning {G}enerative {A}dversarial {N}etwork {A}ttacks,'' in {\em Proc. Int. Symp. Commun. Syst., Netw. Digit. Signal Process. (CSNDSP)}, pp.~503--508, 2022.

\bibitem{Shi2021}
Y.~Shi, K.~Davaslioglu, and Y.~E. Sagduyu, ``{G}enerative {A}dversarial {N}etwork in the {A}ir: {D}eep {A}dversarial {L}earning for {W}ireless {S}ignal {S}poofing,'' {\em IEEE Trans. Cogn. Commun. Netw.}, vol.~7, no.~1, pp.~294--303, 2021.

\bibitem{WangINFOCOM2020}
N.~Wang, L.~Jiao, P.~Wang, W.~Li, and K.~Zeng, ``{M}achine {L}earning-based {S}poofing {A}ttack {D}etection in {M}mwave 60{GHz} {IEEE} 802.11ad {N}etworks,'' in {\em Proc. IEEE INFOCOM}, pp.~2579--2588, 2020.

\bibitem{Wang_TETCI2018}
M.~Wang, S.~Abdelfattah, N.~Moustafa, and J.~Hu, ``{D}eep {G}aussian {M}ixture-{H}idden {M}arkov {M}odel for {C}lassification of {EEG} {S}ignals,'' {\em IEEE Transactions on Emerging Topics in Computational Intelligence}, vol.~2, no.~4, pp.~278--287, 2018.

\end{thebibliography}

\end{document}